\documentclass[10pt,conference]{IEEEtran}
\usepackage{amsmath}
\usepackage{amssymb}
\usepackage{epsfig}
\usepackage{subfig}
\usepackage{verbatim}
\usepackage{paralist}
\usepackage{tabulary}
\usepackage{cite}

\DeclareSymbolFont{calsymbols}{OMS}{cmsy}{m}{n} \DeclareSymbolFontAlphabet{\mathcal}{calsymbols}

\def\E{\mathop\mathbb{E}\nolimits}

\begin{document}

\newtheorem{theorem}{Theorem} \newtheorem{lemma}{Lemma}
\newtheorem{definition}{Definition} \newtheorem{corollary}{Corollary}
\newtheorem{proposition}{Proposition}
\newtheorem{example}{Example}

\IEEEoverridecommandlockouts

\title{{\huge Delay Constrained Scheduling over Fading Channels: \\Optimal Policies for Monomial Energy-Cost Functions}\thanks{The work of J.~Lee is supported by a Motorola Partnership in Research Grant.}}

\author{\authorblockN{Juyul Lee and Nihar Jindal}\\
\authorblockA{Department of Electrical and Computer Engineering\\
University of Minnesota\\
E-mail: \{juyul, nihar\}@umn.edu}}

\maketitle

\begin{abstract}
	A point-to-point discrete-time scheduling problem of transmitting $B$ information bits within $T$ hard delay deadline slots is considered assuming that the underlying energy-bit cost function is a convex monomial. The scheduling objective is to minimize the expected energy expenditure while satisfying the deadline constraint based on information about the unserved bits, channel state/statistics, and the remaining time slots to the deadline. At each time slot, the scheduling decision is made without knowledge of future channel state, and thus there is a tension between serving many bits when the current channel is good versus leaving too many bits for the deadline.
	Under the assumption that no other packet is scheduled concurrently and no outage is allowed, we derive the optimal scheduling policy. Furthermore, we also investigate the dual problem of maximizing the number of transmitted bits over $T$ time slots when subject to an energy constraint.
\end{abstract}

\section{Introduction}

An opportunistic scheduling policy that adapts to the time-varying behavior of a wireless channel can achieve energy-efficient communication on the average in a long-term perspective. However, this opportunistic approach may not be appropriate for  short-term deadline constrained traffic. This paper considers scheduling a packet over a finite time horizon while efficiently adapting to wireless (fading) channel variations and taking care of the deadline constraint.

Our primal problem setting is the minimization of energy expenditure subject to a hard deadline constraint (i.e., a packet of $B$ bits must be scheduled within finite $T$ discrete-time slots) assuming that the scheduler has \emph{causal} knowledge of the channel state information (CSI). Causal CSI means that the scheduler knows the past and current CSI perfectly, but does not know future CSI. The scheduler is then required to make a decision at each time slot given the number of unserved bits, the number of slots left before the deadline, and causal CSI, in order to minimize the total energy expenditure. At each time slot, the scheduler deals with the tension between serving more bits when the channel is good and leaving too many bits to the end. Likewise, we consider the dual (scheduling over a finite time-horizon) problem of maximizing the transmitted bits subject to a finite energy constraint. We also briefly discuss scheduling problems when the CSI is available non-causally. We assume that no other packet is scheduled simultaneously and the hard delay deadline must be met (i.e., no outage is allowed). These finite-time horizon scheduling problems can be applicable to regularly arriving packets with hard delay deadlines, e.g., VoIP and video streaming.

Delay constrained scheduling over fading channel has been studied for various traffic models and delay constraints. Uysal-Biyikoglu and El Gamel \cite{UysalBiyikoglu_IT04} considered scheduling random packet arrivals over a fading channel and thus adapt (transmit power/rate) to both the channel state and queue state, and generally try to minimize \emph{average} delay. Many references can be found in \cite{UysalBiyikoglu_IT04}. Most cases do not admit analytical closed-form solution for causal (or online) scheduling. Instead, they proposed causal algorithms with heuristic modifications from non-causal (offline) policies. References \cite{Fu_WC06, Negi_IT02, Lee_WC08_accepted} take a slightly different perspective: single packet scheduling (no queue) with a hard delay deadline rather than an average delay constraint. 

The subject of this paper is the single-packet scheduling problem of \cite{Fu_WC06} specialized to the case where the required energy $E$ to transmit $b$ bits under channel state $g$ is governed by a convex monomial function, i.e., $E=b^n/g$, where $n$ denotes the monomial order.
The biggest advantage of using this monomial cost function is that it yields closed-form solutions in various scenarios, unlike the Shannon-cost function setting described in \cite{Lee_WC08_accepted}. As a result, it provides intuition on the interplay between the monomial order, delay deadline, and the channel states so that it ultimately suggests general ideas for a more general energy-cost function.
Although the monomial cost does not hold for operating at capacity in an AWGN channel, according to Zafer and Modiano \cite{Zafer_WITA07} and their reference \cite{Neely_INFOCOM03}, there is a practical modulation scheme that exhibits an energy-bit relation that can be well approximated by a monomial. Actually, Zafer and Modinano \cite{Zafer_WITA07} considered the same problem but for a continuous-time Markov process channel in continuous-time scheduling, i.e., the scheduler can transmit at any time instant rather than discrete slotted time. Although they provided a solution in the form of a set of differential equations, it is not possible to give a closed-form solution. On the other hand, we are able to derive a closed-form description of the optimal scheduler for the simpler block fading model (note that the continuous model is somewhat incompatible with block fading).

In this paper, we derive optimal scheduling policies for delay-constrained scheduling when the energy-bit cost is a convex monomial function. We also investigate the dual problem of maximizing the number of bits to transmit with a finite energy budget over a finite time horizon. In all cases, we are able to find analytical expressions that are functions of the queue state variables (energy state for the dual problem), current channel state and a quantity related to the fading distribution.

The resulting optimal schedulers determine the ratio of the number of bits to be allocated in the current slot to the deferred bits. For example, the optimal scheduling ratio of the number of bits to serve $b_t$ (from the remaining $\beta_t$ bits) at slot $t$ ($t$ denotes the number of remaining slots to the deadline) to the number of bits to defer $(\beta_t-b_t)$ for the primal energy minimization problem is given by
\begin{equation}
	b_t : (\beta_t-b_t) = g_t^{\frac{1}{n-1}} : \eta_{n,t},
\end{equation}
where $n$ is the order of monomial cost function, $g_t$ denotes the current channel state, and $\eta_{n,t}$ denotes a statistical quantity determined by the channel distribution and the number of remaining slots $t$. It will be shown later that $\eta_{n,t}$ is increasing with respect to $t$. 
If $\eta_{n,t}$ is small, $b_t\approx\beta_t$. However, as $\eta_{n,t}$ term increases, $b_t$ gets more affected by the channel state $g_t^{\frac{1}{n-1}}$. This suggests that the scheduler behaves very opportunistically when the deadline is far away ($t$ large) but less so as the deadline approaches, since $\eta_{n,t}$ is an increasing function of $t$. 

\section{Primal Problem: Energy Minimization}
\label{sec:problem_energy_min}

We consider the scheduling of a packet of $B$ bits in $T$ discrete time slots over a wireless channel as illustrated in Fig.~\ref{fig:scheduling_simple}.
\begin{figure}
	\centering
	\includegraphics[width=0.48\textwidth]{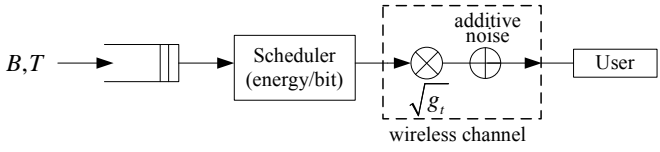}
	\caption{Point-to-point delay constrained scheduling}
	\label{fig:scheduling_simple}
\end{figure}
The scheduler determines the number of bits to allocate at each time slot using the fading realization/statistics to minimize the total transmit energy while satisfying the delay deadline constraint.
To make the scheduling problem tractable, we assume that no other packets are to be scheduled simultaneously and that no outage is allowed.

Throughout the paper, we use the following notations:
\begin{compactitem}
	\item $T$: the number of time slots that a packet of $B$ bits must be transmitted within; the delay deadline.
	\item $t$: discrete-time index in descending order (starting at $t=T$ and all the down to $t=1$); $t$ denotes the number of remaining slots.
	\item $g_t$: the channel state (in power unit) at time slot $t$.
	\item $b_t$: the number of transmitted bits in slot $t$ (there is no integer constraint on $b_t$).
	\item $\beta_t$: the remaining bits at the beginning of time slot $t$; the queue state.
	\item $E_t$: the energy cost in time slot $t$.
\end{compactitem}

The channel states $\{g_t\}_{t=1}^T$ are assumed to be independently and identically distributed (i.i.d.).  If the scheduler has only causal knowledge of the channel state (i.e., at slot $t$, the scheduler knows $g_T, g_{T-1},\cdots,g_t$ but does not know $g_{t-1},g_{t-2},\cdots, g_1$), we refer to this as \emph{causal scheduling}. If the scheduler has non-causal knowledge of the channel state in advance (i.e., at slot $T$, the scheduler knows $g_T,g_{T-1},\cdots, g_1$), we refer to it as \emph{non-causal scheduling}. This paper mainly deals with causal scheduling problems.

In this paper, we assume that the energy expenditure $E_t$ is inversely proportional\footnote{The $1/g_t$ dependence is due to the fact that the received energy is the product of the transmitted energy $E_t$ and the channel state $g_t$. Note, however, that any other decreasing function of $g_t$ could be considered by simply performing a change of variable on $g_t$.} to the channel state $g_t$ and is related to the transmitted bits $b_t$ by a monomial function:
\begin{equation} \label{eq:Et_exp}
	E_t(b_t,g_t;n) = \frac{b_t^n}{g_t},
\end{equation}
where $n$ denotes the order of monomial. If $n=1$, the resulting optimization becomes a linear program and thus a ``one-shot''  policy is optimal \cite{Lee_asymptotic_unpublished}. We assume that $n>1$ (to be convex) and $n\in \mathbb{R}$ ($n$ is not necessarily an integer), where $\mathbb{R}$ denotes the real number set. A practical modulation scheme that exhibits a monomial energy-cost behavior was illustrated in \cite{Zafer_WITA07}, where the monomial order is $n=2.67$.

A scheduler is a sequence of functions $\{b_t(\beta_t,g_t)\}_{t=1}^T$ with $0\le b_t\le \beta_t$.  For causal scheduling, $b_t$ depends only on the current channel state $g_t$ and not on the past and future states because of the i.i.d.~assumption and causality\footnote{The i.i.d.~assumption makes us ignore the past CSI $g_T, g_{T-1},\cdots, g_{t+1}$ and the causality does not allow to exploit the future CSI $g_{t-1},g_{t-2},\cdots, g_1$. As a result, the decision at each time slot should be made based only on the current CSI $g_t$, i.e., $b_t(\beta,g_t)$ instead of $b_t(\beta,g_T,\cdots,g_1)$.}. The optimal scheduler is determined by minimizing the total expected energy cost: \vspace{-10pt}
\begin{eqnarray} 
	&\min\limits_{b_T,\cdots,b_1}& \E\left[\sum_{t=1}^T E_t(b_t,g_t;n) \right] \nonumber \\
	&\text{subject to}& \sum_{t=1}^T b_t =B \label{eq:min_E_sums} \\
	&& b_t\ge 0,\qquad \forall t, \nonumber
\end{eqnarray}
where $\E$ denotes the expectation operator. 

\section{Causal Energy Minimization Scheduling}

As done in \cite{Fu_WC06}\cite{Lee_WC08_accepted}, a sequential formulation of the optimal causal scheduling of \eqref{eq:min_E_sums} can be established by introducing a state variable $\beta_t$ as in standard dynamic programming \cite{Bertsekas_DP1_Book05}. As defined in Section \ref{sec:problem_energy_min}, $\beta_t$ denotes the remaining bits that  summarizes the bit allocation up until the previous time step. At time step $t$, $g_{t-1},\cdots,g_1$ are unknown but $g_t$ is known. Thus, the optimization \eqref{eq:min_E_sums} becomes:
\begin{equation*}
	\min_{0\le b_t\le\beta_t} \left(E_t(b_t,g_t;n)+ \E\left[\sum_{s=1}^{t-1} E_s(b_s,g_s;n)\Bigg\vert b_t\right]\right),\; t\ge 2.
\end{equation*}
With \eqref{eq:Et_exp}, we obtain the following DP:
\begin{equation} \label{eq:Jt_csl}
	J_t^\text{csl}(\beta_t,g_t;n) =
	\begin{cases}
		\min\limits_{0\le b_t\le \beta_t} \left(\frac{b_t^n}{g_t}+\bar{J}_{t-1}^\text{csl}(\beta_t-b_t;n)\right), & t\ge 2\\
		\frac{\beta_1^n}{g_1}, & t=1,
	\end{cases}
\end{equation}
where the first term $\frac{b_t^n}{g_t}$ denotes the current energy cost and the second term $\bar{J}_{t-1}^\text{csl}(\beta; n)=\E_{g}[J_{t-1}^\text{csl}(\beta,g; n)]$ denotes the cost-to-go function, which is the expected future energy cost (because future channel states are unknown, only expectations can be considered) to serve $\beta$ bits in $(t-1)$ slots if the optimal control policy is used at each future step. Thus, the optimal bit allocation is determined by balancing the current energy cost and the expected future energy cost. Because of the hard delay constraint, all the unserved bits must be served at $t=1$ regardless of the channel condition, i.e., $b_1=\beta_1$ and thus the resulting energy cost is given by $\frac{\beta_1^n}{g_1}$. This dynamic optimization can be solved:
\begin{theorem} \label{thm:bt_csl}
	The optimal solution to the causal energy minimization scheduling problem \eqref{eq:Jt_csl} is given by
	\begin{equation} \label{eq:bt_csl}
	\boxed{ b_t^\text{csl}(\beta_t, g_t;n) = \begin{cases} \beta_t \left( \frac{(g_t)^{\frac{1}{n-1}}}{(g_t)^{\frac{1}{n-1}}+\left(\frac{1}{\xi_{n,t-1}}\right)^{\frac{1}{n-1}}}\right), & t\ge 2\\
	\beta_1,& t=1, \end{cases}}
	\end{equation}
	where the constants $\xi_{n,t}$ are determined as:
\begin{equation} \label{eq:xi}
	\xi_{n,t} = \begin{cases}
		\E\left[\left(\frac{1}{(g_t)^{\frac{1}{n-1}}+ (1/\xi_{n,t-1})^{\frac{1}{n-1}}}\right)^{n-1}\right],& t\ge 2,\\
		\E\left[\frac{1}{g}\right], & t=1,
	\end{cases}
\end{equation}
	and the expected energy cost is given by
\begin{equation} \label{eq:barJt_csl}
	\bar{J}_{t}^\text{csl}(\beta;n) = \beta^n \xi_{n,t},\quad t=1,2,\cdots.
\end{equation}
\end{theorem}
\vspace{5pt}
\begin{proof}
	We use mathematical induction to find $b_t^\text{csl}(\cdot,\cdot;n)$ and $\bar{J}_t^\text{csl}(\cdot;n)$. At $t=1$, \eqref{eq:bt_csl} and \eqref{eq:barJt_csl} are true by definition. 	If we suppose that \eqref{eq:barJt_csl} is true for $t-1$, the optimization \eqref{eq:Jt_csl} becomes
	\begin{equation} \label{eq:Jt_temp}
		J_t^\text{csl}(\beta_t,g_t;n) = \min_{0\le b_t\le \beta_t} \left(\frac{b_t^n}{g_t} + (\beta_t-b_t)^n \xi_{n,t-1}\right),
	\end{equation}
	whose solution is obtained by differentiating the objective and setting to zero to result in \eqref{eq:bt_csl}.
	Substituting \eqref{eq:bt_csl} into \eqref{eq:Jt_temp} and then taking expectation with respect to $g_t$, we obtain \eqref{eq:barJt_csl}. Therefore, the result follows by induction. 
\end{proof}

The scheduling function \eqref{eq:bt_csl} can be intuitively explained in the following way. 
The ratio of the number of allocated bits $b_t$ to the number of deferred bits $(\beta_t-b_t)$ is equal to the ratio of $g_t^{\frac{1}{n-1}}$ to $\left(1/\xi_{n,t-1}\right)^\frac{1}{n-1}$, i.e.,
\begin{equation}
	b_t : \underbrace{(\beta_t-b_t)}_{\beta_{t-1}} = (g_t)^{\frac{1}{n-1}} : \underbrace{\eta_{n,t}}_{\text{threshold}}
\end{equation}
where $\eta_{n,t}=\left(1/\xi_{n,t-1}\right)^\frac{1}{n-1}$. As expected, the optimal scheduler is \emph{opportunistic} in that the number of transmitted bits are proportional to the channel quality. Furthermore, the thresholds $\eta_{n,t}$ are increasing in $t$ (shown later) which implies that the scheduler is more \emph{selective} when the delay deadline is far away (large $t$). When the deadline is far away, the scheduler transmits a large fraction of the unserved bits only when the channel state is very good; because many slots remain until the deadline, there is still a good chance of seeing a very good channel state. On the other hand, as the deadline approaches (small $t$) the scheduler is still opportunistic but must become less selective because only a few opportunities for good channel states remain before the deadline is reached.

Figure \ref{fig:xi_nt} illustrates $\eta_{n,t}\left(= (1/\xi_{n,t})^{\frac{1}{n-1}}\right)$ and $\xi_{n,t}$ for a truncated exponential distribution.
\begin{figure}
	\centering
	\subfloat[$(1/\xi_{n,t})^{\frac{1}{n-1}}$]{\includegraphics[width=0.49\textwidth]{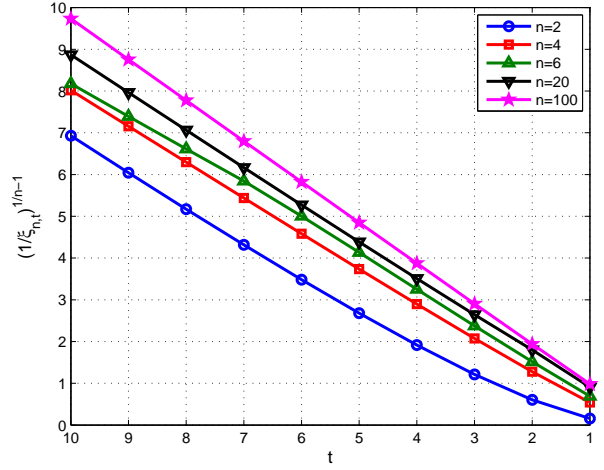}}
	\subfloat[$\xi_{n,t}$]{\includegraphics[width=0.49\textwidth]{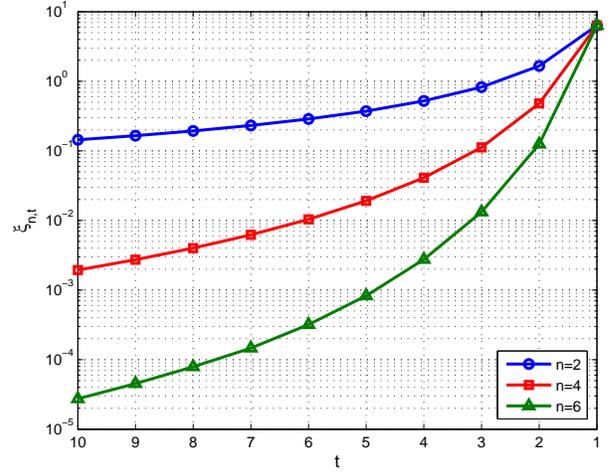}}
	\caption{$\xi_{n,t}$ for the truncated exponential random variable $g$ with threshold $0.001$, i.e., $f(g)=e^{-(g-0.001)}$ if $g\ge 0.001$ and $f(g)=0$ if $g<0.001$, where $f$ denotes the PDF of $g$.}
	\label{fig:xi_nt}
\end{figure}
As can be seen in Fig.~\ref{fig:xi_nt}a, $\eta_{n,t}$ increases with respect to $t$ and this can be shown analytically: \vspace{-15pt}
\begin{equation} \label{eq:xi_decrease}
	\begin{split}
		\xi_{n,t} &= \E\left[\left(\frac{1}{g_t^{\frac{1}{n-1}}+\left(\frac{1}{\xi_{n,t-1}}\right)^{\frac{1}{n-1}}}\right)^{n-1}\right] \\
		& \le \E\left[\left(\frac{1}{\left(\frac{1}{\xi_{n,t-1}}\right)^{\frac{1}{n-1}}}\right)^{n-1}\right]=\xi_{n,t-1}
	\end{split}
\end{equation}
where the inequality is due to $g_t \ge 0$. This shows the delay-limited opportunistic behavior mentioned before. From \eqref{eq:barJt_csl}, the value $\xi_{n,t}$ denotes the expected energy cost for a unit bit, i.e., $\beta_t=1$. Thus, $\xi_{n,t}$, as illustrated in Fig.~\ref{fig:xi_nt}b, shows how much the expected energy unit cost (for transmitting one bit) can be reduced as the time span increases.

Another interesting fact is that the policy \eqref{eq:bt_csl} utilizes\footnote{A time slot $t$ is called \emph{utilized} if a positive bit is scheduled, i.e., $b_t>0$.} all the time slots. This is because both $(g_t)^{\frac{1}{n-1}}$ and $\left(1/\xi_{n,t-1}\right)^\frac{1}{n-1}$ are always positive for typical fading distributions. 
For the Shannon cost function problem \cite{Lee_WC08_accepted}, however, there exist time slots that are not utilized depending on the values of $B$ and $T$. This does not admit an analytical solution because the associated cost-to-go function takes a complicated form. 

\subsection{Special Cases}

In this subsection, we examine the optimal policy \eqref{eq:bt_csl} for two values of $n$: $n=2$ and $n\to\infty$.

\subsubsection{Quadratic Cost ($n=2$)}

By substituting $n=2$ in \eqref{eq:bt_csl} and \eqref{eq:xi}, we have
\begin{equation}
	b_t^\text{csl}(\beta_t,g_t; n=2) = \beta_t \left(\frac{g_t}{g_t+\frac{1}{\xi_{2,t-1}}}\right),
\end{equation}
where
\begin{equation}
	\xi_{2,t} = \begin{cases}
		\E\left[\frac{1}{g_t + \frac{1}{\xi_{2,t-1}}}\right],& t\ge 2,\\
		\E\left[\frac{1}{g}\right],& t=1.
	\end{cases}
\end{equation}
Thus, the allocated bits $b_t$ and the deferred bits $(\beta_t-b_t)$ have the same ratio with $g_t$ and $1/\xi_{2,t-1}$. 

\subsubsection{Infinite Order Cost ($n=\infty$)}

We examine the limiting behavior of the scheduling policy \eqref{eq:bt_csl} as $n\to\infty$. First, we observe that
\begin{lemma} \label{lem:lim_n_xi}
	\begin{equation} \label{eq:lim_n_xi}
		\lim_{n\to\infty}	\left(\frac{1}{\xi_{n,t}}\right)^{\frac{1}{n-1}} = t.
	\end{equation}
\end{lemma}
\begin{proof}
This can be shown by the induction. When $t=1$, \eqref{eq:lim_n_xi} holds trivially. If we suppose \eqref{eq:lim_n_xi} holds for $t-1$, then
\begin{equation}
	\begin{split}
	&\lim_{n\to\infty} \left(\xi_{n,t}\right)^{\frac{1}{n-1}} \\
	&= \lim_{n\to\infty} \left(\E\left[\frac{1}{\left(g_{t-1}^{\frac{1}{n-1}}+\left(\frac{1}{\xi_{n,t-1}}\right)^{\frac{1}{n-1}}\right)^{n-1}}\right]\right)^{\frac{1}{n-1}} \\
	&= \frac{1}{t}, 
	\end{split}
\end{equation}
where the last equality is due $\lim_{n\to\infty}\left(\E[\phi^n]\right)^{\frac{1}{n}}= \text{Max} \phi$ and $\text{Max} \phi$ denotes the ``effective upper bound'' of $\phi$ (see Chap.~6 in \cite{Hardy_Book01} for mathematical technicality). Hence, the induction follows.
\end{proof}
Figure \ref{fig:xi_nt}a illustrates the values of $(1/\xi_{n,t})^{\frac{1}{n-1}}$ for the truncated exponential variable.
This shows that $\left(1/\xi_{n,t}\right)^\frac{1}{n-1}$ is increasing linearly with respect to $t$ for large $n$, which agrees with Lemma \ref{lem:lim_n_xi}.

With the limit in Lemma \ref{lem:lim_n_xi}, we can immediately reach the simplified scheduling policy summarized below:
\begin{theorem} \label{thm:bt_csl_limit}
	As $n\to\infty$, the scheduling policy \eqref{eq:bt_csl} becomes the equal-bit scheduler, i.e.,
	\begin{equation}
		b_t^\text{csl}(\beta_t,g_t; n=\infty) = \frac{\beta_t}{t},\quad t=1,2,\cdots.
	\end{equation}
\end{theorem}
That is, when the order of monomial cost function tends to infinity, scheduling equal number of bits at every slot regardless of the channel state becomes the optimal policy. 
Note that we considered only monomial orders $n>1$ in the derivation, as when $n=1$, the optimal policy is the one-shot policy \cite{Lee_asymptotic_unpublished}, which completely depends on the channel state. 
From these two extreme cases, we can deduce that the effect of channel state on the scheduling function decreases as the order of monomial cost function increases, or in other words the optimal scheduler becomes less opportunistic as the monomial order $n$ increases.

\section{Dual Problem: Rate Maximization}

Thus far, we have considered problems of minimizing energy expenditure to transmit fixed $B$ information bits in a finite time horizon $T$. It is of interest to consider the dual of this, i.e., maximizing the number of bits transmitted with a finite energy $E$ over a finite time horizon $T$. We refer to this as the \emph{dual scheduling problem}, while referring to the original problem as the \emph{primal scheduling problem}. Negi and Cioffi \cite{Negi_IT02} considered this dual problem for the Shannon energy-bit cost function and provided solutions in DP, but not in closed form. In this work, we investigate this dual scheduling problem and obtain the optimal closed-form solution for monomial cost functions.

Since the energy-bit function is assumed to be \eqref{eq:Et_exp}, the associated bit-energy cost function is given by inverting: \begin{equation}
	b_t = \left(g_t E_t\right)^{\frac{1}{n}}.
\end{equation}
Then the dual problem is given by
\begin{eqnarray}
	&\max\limits_{E_T,\cdots,E_1}& \E\left[\sum_{t=1}^T \left(g_t E_t\right)^{\frac{1}{n}}\right] \nonumber \\
	&\text{subject to}& \sum_{t=1}^T E_t = E \label{eq:dual} \\
	&& E_t\ge 0,\qquad \forall t. \nonumber
\end{eqnarray}

To derive a DP for causal dual scheduling, we introduce a state variable $\mathcal{E}_t$ that denotes the remaining energy at slot $t$.
Thus, the optimization \eqref{eq:dual} can be formulated as
\begin{equation} \label{eq:Wt_csl}
	\begin{split}
	&W_t^\text{csl}(\mathcal{E}_t,g_t;n) = \\
	&\begin{cases}
		\max\limits_{0\le E_t\le \mathcal{E}_t} \left( (g_t E_t)^{\frac{1}{n}}+\bar{W}_{t-1}^\text{csl}(\mathcal{E}_t-E_t;n)\right),&  t\ge 2\\
		(g_1 \mathcal{E}_1)^{\frac{1}{n}},&  t=1,
	\end{cases}
	\end{split}
\end{equation}
where $\bar{W}_{t-1}^\text{csl}(\mathcal{E};n)=\E_g[W_{t-1}^\text{csl}(\mathcal{E},g;n)]$ denotes the cost-to-go function for the dual scheduling problem. This dynamic optimization \eqref{eq:Wt_csl} can be solved similar to the primal problem and its optimal solution is summarized as follows:
\begin{theorem}
	The optimal causal rate maximization scheduling \eqref{eq:Wt_csl} is given by
	\begin{equation} \label{eq:Et_csl}
	\boxed{
		E_t^\text{csl}(\mathcal{E}_t,g_t;n) 
		= \begin{cases} \mathcal{E}_t \left(\frac{(g_t)^{\frac{1}{n-1}}}{(g_t)^{\frac{1}{n-1}} + (\zeta_{n,t-1})^{\frac{1}{n-1}}}\right), & t\ge 2,\\
		\mathcal{E}_1,& t=1,
		\end{cases}}
	\end{equation}
	where \vspace{-10pt}
	\begin{equation}
		\zeta_{n,t} = \begin{cases}
		\left(\E\left[\left((g_t)^{\frac{1}{n-1}} + (\zeta_{n,t-1})^{\frac{1}{n-1}}\right)^{\frac{n-1}{n}}\right]\right)^n, & t\ge 2,\\
		\left(\E[g^{\frac{1}{n}}]\right)^n, & t=1.
		\end{cases}
	\end{equation}
\end{theorem}
The optimal energy scheduler \eqref{eq:Et_csl} has very similar interpretation with the optimal bit scheduler \eqref{eq:bt_csl} from their scheduling formulations. That is, the ratio of the amount of energy to schedule $E_t$ to the amount of energy to defer $(\mathcal{E}_t-E_t)$ is equal to the ratio of $g_t^{\frac{1}{n-1}}$ to ${\zeta}_{n,t-1}^{\frac{1}{n-1}}$, and thus, the similar delay-limited opportunistic scheduling interpretation can be applied. Notice that the quantities $\zeta_{n,t}$ and $\xi_{n,t}$ are different.

\section{Non-Causal Scheduling}

This section briefly considers the case where the scheduler has knowledge of the channel states non-causally in advance, i.e., $g_T, g_{T-1},\cdots, g_1$ are known at $t=T$.

\subsection{Energy Minimization Scheduling}

In this non-causal setting, the optimization \eqref{eq:min_E_sums} is simply given by \vspace{-10pt}
\begin{equation} \label{eq:minE_ncsl}
	\min_{b_T,\cdots,b_1} \sum_{t=1}^T \frac{b_t^n}{g_t} 
\end{equation}
subject to $\sum_{t=1}^T b_t=B$ and $b_t\ge 0$ for all $t$. This is a convex optimization and can be solved as:
\begin{theorem} \label{thm:bt_ncsl}
	The optimal non-causal scheduling to \eqref{eq:minE_ncsl} is given by \vspace{-10pt}
	\begin{equation} \label{eq:bt_ncsl}
		\boxed{b_t^\text{ncsl}(\beta_t,g_t; n) = \beta_t \frac{g_t^{\frac{1}{n-1}}}{\sum_{s=1}^t g_s^{\frac{1}{n-1}}}. }
	\end{equation}
\end{theorem}
\begin{proof}
	The standard Lagrangian method \cite{Boyd_Book04} yields the solution: \vspace{-10pt}
	\begin{equation} \label{eq:bt_ncsl_withB}
		b_t^\text{ncsl} = B \frac{g_t^{\frac{1}{n-1}}}{\sum_{s=1}^T g_s^{\frac{1}{n-1}}}.
	\end{equation}
	If we express this solution with the queue state variable $\beta_t$, we obtain the result.
\end{proof}
The scheduling policy \eqref{eq:bt_ncsl} can be interpreted with the ratio argument as with the causal cases, i.e., 
\begin{equation} \label{eq:bt_ncsl_ratios}
	b_T^\text{ncsl} : b_{T-1}^\text{ncsl} : \cdots : b_1^\text{ncsl} 
	= g_T^{\frac{1}{n-1}} : g_{T-1}^{\frac{1}{n-1}} : \cdots : g_1^{\frac{1}{n-1}}.
\end{equation}

\subsection{Rate Maximization Scheduling}

Similarly we can fomulate  the non-causal rate maximition as \vspace{-5pt}
\begin{equation} \label{eq:maxB_ncsl}
	\max \sum_{t=1}^T (g_t E_t)^{\frac{1}{n}},
\end{equation}
subject to $\sum_{t=1}^T E_t = E$ and $E_t\ge 0$ for all $t$. 
\begin{theorem}
	The optimal non-causal scheduling to \eqref{eq:maxB_ncsl} is given by \vspace{-5pt}
	\begin{equation} \label{eq:Et_ncsl}
		\boxed{E_t^\text{ncsl} = \mathcal{E}_t\frac{g_t^{\frac{1}{n-1}}}{\sum_{s=1}^{t} g_s^{\frac{1}{n-1}}}.}
	\end{equation}
\end{theorem}
Like \eqref{eq:bt_ncsl_ratios}, we can also observe that
\begin{equation}
	E_T^\text{ncsl} : E_{T-1}^\text{ncsl} : \cdots : E_1^\text{ncsl} 
	= g_T^{\frac{1}{n-1}} : g_{T-1}^{\frac{1}{n-1}} : \cdots : g_1^{\frac{1}{n-1}},
\end{equation}
and thus, we obtain \vspace{-5pt}
\begin{equation}
	\frac{b_t^\text{ncsl}}{B}=\frac{E_t^\text{ncsl}}{E}.
\end{equation}
This implies that the optimal bit distribution ratio during the $T$ slots for the primal problem is identical to the energy distribution ratio for the dual problem.

\section{Conclusion}

We have investigated the problem of bit/energy scheduling over a finite time duration assuming that the energy-bit cost function is a monomial. In both the primal (minimizing energy expenditure subject to a bit constraint) scheduling and the dual (maximizing bit transmission under an energy constraint) scheduling problem, we derived closed-form scheduling functions. The optimal bit/energy allocations are determined by the ratio of $g_t^{\frac{1}{n-1}}$ and a channel statistical quantity. From the monotonicity of this statistical quantity, we interpreted that the optimal scheduler behaves more opportunistically in the initial time steps and less so as the deadline approaches.

\bibliographystyle{IEEEtran}
\bibliography{monomial}
\end{document}